\documentclass[11pt,a4paper]{article}
\usepackage{amsmath}\usepackage{epsf,amsfonts,amsthm}\usepackage{amscd,amssymb}
\usepackage{xcolor,epic,eepic}\usepackage{epsfig}
\usepackage{fontenc,indentfirst, delarray,amsfonts,amsmath,amssymb}
\usepackage{rotating}
\usepackage{mathdots}
\usepackage[T1]{fontenc}
\usepackage[matrix,arrow,curve]{xy}
\usepackage{amsmath}
\usepackage{amssymb}
\usepackage{amsthm}
\usepackage{amscd}
\usepackage{amsfonts}
\usepackage{graphicx}
\usepackage{fancyhdr}
\usepackage{dsfont,texdraw}
\usepackage{amsmath}\usepackage{epsf,amsfonts,amsthm}\usepackage{amscd,amssymb}
\usepackage{xcolor,epic,eepic}\usepackage{epsfig}
\usepackage{fontenc,indentfirst, delarray,amsfonts,amsmath,amssymb}
\usepackage{rotating}
\usepackage{amscd}
\usepackage{amsmath}
\usepackage{amsthm}
\usepackage{mathrsfs}
\usepackage{latexsym}
\usepackage{amssymb}
\usepackage{amsfonts}
\theoremstyle{plain}

\usepackage{tikz}
\usepackage{tikz-cd}
\usetikzlibrary{arrows,chains,matrix,positioning,scopes,snakes}
\makeatletter
\tikzset{join/.code=\tikzset{after node path={%
\ifx\tikzchainprevious\pgfutil@empty\else(\tikzchainprevious)%
edge[every join]#1(\tikzchaincurrent)\fi}}}
\makeatother

\tikzset{>=stealth',every on chain/.append style={join},
         every join/.style={->}}
\tikzset{
    >=stealth',
    punkt/.style={
           rectangle,
           rounded corners,
           draw=black, very thick,
           text width=6.5em,
           minimum height=2em,
           text centered},
    pil/.style={
           ->,
           thick,
           shorten <=2pt,
           shorten >=2pt,}
}

\usepackage[T1]{fontenc}
\newcommand{\bee}{\begin{enumerate}}
\newcommand{\eee}{\end{enumerate}}
\newcommand{\benn}{\begin{equation*}}
\newcommand{\eenn}{\end{equation*}}
\newcommand{\be}{\begin{equation}}
\newcommand{\ee}{\end{equation}}
\newcommand{\bean}{\begin{eqnarray}}
\newcommand{\eean}{\end{eqnarray}}
\newcommand{\bea}{\begin{eqnarray*}}
\newcommand{\eea}{\end{eqnarray*}}

\newcommand{\p}{\partial}

\newcommand{\Ci}{C^{\infty}}

\newcommand{\N}{\mathbb{N}}
\newcommand{\Z}{\mathbb{Z}}
\newcommand{\R}{\mathbb{R}}

\newcommand{\C}{\mathbb{C}}

\newcommand{\op}[1]{\!\!\mathop{\rm ~#1}\nolimits}

\newcommand{\fm}{\frak{m}_m}
\newcommand{\cN}{{\cal N}}
\newcommand{\cF}{{\cal F}}

\mathchardef\za="710B  %\alpha
\mathchardef\zb="710C  %\beta
\mathchardef\zg="710D  %\gamma
\mathchardef\zd="710E  %\delta
\mathchardef\zve="710F %\epsilon
\mathchardef\zz="7110  %\zeta
\mathchardef\zh="7111  %\eta
\mathchardef\zy="7112 %\theta

\mathchardef\zi="7113  %\iota
\mathchardef\zk="7114  %\kappa
\mathchardef\zl="7115  %\lambda
\mathchardef\zm="7116  %\mu
\mathchardef\zn="7117  %\nu
\mathchardef\zx="7118  %\xi
\mathchardef\zp="7119  %\pi
\mathchardef\zr="711A  %\rho
\mathchardef\zs="711B  %\sigma
\mathchardef\zt="711C  %\tau
\mathchardef\zu="711D  %\upsilon
\mathchardef\zf="711E %\phi
\mathchardef\zq="711F  %\chi
\mathchardef\zc="7120  %\psi
\mathchardef\zw="7121  %\omega
\mathchardef\ze="7122  %\varepsilon
\mathchardef\zvy="7123  %\vartheta
\mathchardef\zvw="7124  %\varomega
\mathchardef\zvr="7125 %\varrho
\mathchardef\zvs="7126 %\varsigma
\mathchardef\zvf="7127  %\varphi
\mathchardef\zG="7000  %\Gamma
\mathchardef\zD="7001  %\Delta
\mathchardef\zY="7002  %\Theta
\mathchardef\zL="7003  %\Lambda
\mathchardef\zX="7004  %\Xi
\mathchardef\zP="7005  %\Pi
\mathchardef\zS="7006  %\Sigma
\mathchardef\zU="7007  %\Upsilon
\mathchardef\zF="7008  %\Phi
\mathchardef\zW="700A  %\Omega

\newcommand{\cyclic}{\mathop{\kern0.9ex{{+}
\kern-2.15ex\raise-.25ex\hbox{\Large\hbox{$\circlearrowright$}}}}\limits}

 \newcommand{\cP}{{\cal P}}
 \newcommand{\cC}{{\cal C}}
 \newcommand{\cA}{{\cal A}}
 \newcommand{\cM}{{\cal M}}
 \newcommand{\cD}{{\cal D}}
 \newcommand{\cO}{{\cal O}}
 \newcommand{\cB}{{\cal B}}

\newtheorem{rem}{Remark}
\newtheorem{theo}[rem]{Theorem}
\newtheorem{prop}[rem]{Proposition}
\newtheorem{lem}[rem]{Lemma}
\newtheorem{cor}[rem]{Corollary}
\newtheorem{ex}[rem]{Example}

\newtheorem{defi}[rem]{Definition}

\newcommand{\0}{\otimes}

\newcommand{\id}{\op{id}}

\DeclareMathAlphabet{\mathpzc}{OT1}{pzc}{m}{it}

\newcommand{\Hom}{\mathrm{Hom}}

 \newcommand{\cJ}{\mathcal{J}}

\begin{document}
\title{\bf Functional analytic issues in $\Z_2 ^n$-Geometry}
\date{}
\author{Andrew Bruce and Norbert Poncin\footnote{University of Luxembourg, Mathematics Research Unit, 4364 Esch-sur-Alzette, Luxembourg, andrew.bruce@uni.lu, norbert.poncin@uni.lu}}
\maketitle

\begin{abstract} We show that the function sheaf of a $\Z_2^n$-manifold is a nuclear Fr\'echet sheaf of $\Z_2^n$-graded $\Z_2^n$-commutative associative unital algebras. Further, we prove that the components of the pullback sheaf morphism of a $\Z_2^n$-morphism are all continuous. These results are essential for the existence of categorical products in the category of $\Z_2^n$-manifolds. All proofs are self-contained and explicit.\end{abstract}

\tableofcontents

\section{Introduction}

$\Z_2^n$-Geometry is an emerging framework in mathematics and mathematical physics, which has been introduced in the foundational papers \cite{CGPa} and \cite{COP}. This non-trivial extension of standard supergeometry allows for $\Z_2^n$-gradings, where $$\Z_2^n=\Z_2^{\times n}=\Z_2\times\ldots\times\Z_2\quad\text{and}\quad n\in\N\;.$$ The corresponding $\Z_2^n$-commutation rule for coordinates $(u^A)_A$ with degrees $\deg u^A\in\Z_2^n$ does not use the product of the (obviously defined) parities, but the scalar product $\langle-,-\rangle$ of $\Z_2^n\,$: \be\label{ZCR}u^Au^B=(-1)^{\langle\op{deg}u^A,\op{deg}u^B\rangle}u^Bu^A\;.\ee The definitions of $\Z_2^n$-manifolds and $\Z_2^n$-morphisms are recalled in Section \ref{Summary}. A survey of $\Z_2^n$- Geometry is available in \cite{Pon}.\medskip

The motivations for this new setting originate in both, mathematics and physics.\medskip

In physics, $\Z_2^n$-gradings and the $\Z_2^n$-commutation rule ($n\ge 2$) are used in various contexts. References include \cite{AS}, \cite{AKTT}, \cite{Tol}, and \cite{YJ01}, as well as \cite{CKRS} and \cite{Khoo} (the latter, which can be traced back to \cite{BB}, \cite{DVH}, and \cite{MH}, use the $\Z_2^n$-formalism implicitly). Further, the $\Z_2^n$-commutation rule is not only necessary, but also sufficient: it can be shown \cite{CGP} that any commutation rule, for any finite number $m$ of coordinates, is of the form \eqref{ZCR}, for some $n \ge 2m$.\medskip

In mathematics, well-known algebras are $\Z_2^n$-commutative. This holds for instance for the algebra of quaternions (which is $\Z_2^n$-commutative with $n=3$) and, more generally, for any Clifford algebra $\op{Cl}_{p,q}(\R)$ ($\Z_2^n$-commutative with $n=p+q+1$), as well as for the algebra of Deligne differential forms on a standard supermanifold ($\Z_2^n$-commutative with $n=2$: the first $\Z_2$-degree is induced by the cohomological degree, the second $\Z_2$-degree is the parity).\medskip

Moreover, there exist canonical examples of $\Z^n_2$-manifolds. The local model of these higher or colored supermanifolds is necessarily \cite{Pon} the $\Z_2^n$-commutative algebra $\Ci(x)[[\xi]]$ of formal power series in the coordinates $\xi=(\xi_1,\ldots,\xi_N)$ of non-zero $\Z_2^n$-degree, with coefficients in the smooth functions in the coordinates $x=(x_1,\ldots,x_p)$ of zero $\Z_2^n$-degree. For instance, the tangent bundle of a classical $\Z_2$-manifold or supermanifold $\cM$ can be viewed as a $\Z_2^2$-manifold $T\cM$, in which case the function sheaf is the completion of Deligne differential forms of $\cM$. Actually, the tangent (and cotangent) bundle of any $\Z_2^n$-manifold is a $\Z_2^{n+1}$-manifold. Moreover, the `superization' of any double vector bundle (resp., any $n$-fold vector bundle) is canonically a $\Z_2^2$-manifold (resp., $\Z_2^n$-manifold).\medskip

We expect that a number of applications of $\Z_2^n$-Geometry in physics are based on the integration theory of $\Z_2^n$-manifolds. A first step towards $\Z_2^n$-integration is the $\Z_2^n$-generalization of the Berezinian. This fundamental concept has been constructed in \cite{COP} and is referred to as the $\Z_2^n$-Berezinian. The $\Z_2^n$-integration theory is still under investigation.\medskip

Other applications of $\Z_2^n$-Geometry rely on $\Z_2^n$ Lie groups (generalized super Lie groups) and their actions on $\Z_2^n$-manifolds (which are expected to be of importance in supergravity), on $\Z_2^n$ vector bundles (generalized super vector bundles) and their sections (these are basic objects needed for instance in the study of $\Z_2^n$ Lie algebroids), on the internal Hom in the category of $\Z_2^n$-manifolds (which is of importance in field theory -- $\Z^n$-gradings and $\Z_2^n$-parities), ... All these notions are themselves based on products in the category of $\Z_2^n$-manifolds.\medskip

The investigation in the present paper is motivated by three facts:

\begin{enumerate}

\item the proof of existence of the preceding categorical products uses the content of this paper,

\item the proof of the reconstructions of $\Z_2^n$-manifolds and morphisms between them from global $\Z_2^n$-functions and morphisms between them, is based on this content, and

\item the results of the present paper extend similar results in the standard supercase to the more challenging $\Z_2^n$-context, and constitute an added value of the less detailed study in the supercase, which is given in the Appendix of \cite{CCF}.

\end{enumerate}

\section{$\Z_2^n$-manifolds and their morphisms}\label{Summary}

We denote by $\Z_2^n$ the cartesian product of $n$ copies of $\Z_2\,$. Further, we use systematically the fol\-low\-ing \emph{standard order} of the $2^n$ elements of $\Z_2^n$: first the even degrees are ordered lexicographically, then the odd ones are also ordered lexicographically. For example, $$ \Z_2^3=\{(0,0,0), (0,1,1), (1,0,1), (1,1,0), (0,0,1),(0,1,0), (1,0,0), (1,1,1)\}\;.$$

A $\Z_2^n$-domain has, aside from the usual coordinates $x=(x^1,\ldots,x^p)$ of degree $\deg x^i=0\in\Z_2^n$, also formal coordinates or parameters $\xi=(\xi^1,\ldots,\xi^Q)$ of non-zero degrees $\deg\xi^a\in\Z_2^n$. These coordinates $u=(x,\xi)$ commute according to the generalized sign rule \be\label{SignRule}u^Au^B=(-1)^{\langle\deg u^A,\deg u^B\rangle}u^Bu^A\;,\ee where $\langle-,-\rangle$ denotes the standard scalar product. For instance, $$\langle (0,1,1),(1,0,1)\rangle=1\;.$$ Observe that, in contrast with ordinary $\Z_2$- or super-domains, even coordinates may anticommute, odd coordinates may commute, and even nonzero degree coordinates are not nilpotent. Of course, for $n=1$, we recover the classical situation. We denote by $p$ the number of coordinates $x^i$ of degree 0, by $q_1$ the number of coordinates $\xi^a$ which have the first non-zero degree of $\Z_2^n$, and so on. We get that way a tuple $\mathbf{q}=(q_1,\ldots,q_N)\in\N^N$ with $N:=2^{n}-1$. The dimension of the considered $\Z_2^n$-domain is then given by $p|\mathbf{q}$. Clearly the $Q$ above is the sum $|\mathbf{q}|=\sum_{i=1}^Nq_i$.
\medskip

We recall the definition of a $\Z_2^n$-manifold.

\begin{defi}
A \emph{locally $\Z_2^n$-ringed space} is a pair $(M,\cO_M)$ made of a topological space $M$ and a sheaf of $\Z_2^n$-graded $\Z_2^n$-commutative $(\,$in the sense of \eqref{SignRule}$\,)$ associative unital $\R$-algebras over it, such that at every point $m\in M$ the stalk $\cO_{M,m}$ is a local graded ring.\smallskip

A smooth \emph{$\Z_2^n$-manifold} of dimension $p|\mathbf{q}$ is a \emph{locally $\Z_2^n$-ringed space} $\cM=(M,\cO_M)$, which is locally isomorphic to the smooth $\Z_2^n$-domain $\R^{p|\mathbf{q}}:=(\R^p,\Ci_{\R^p}[[\xi]])$, and whose underlying topological space $M$ is second-countable and Hausdorff. Sections of the structure sheaf $\,\C^{\infty}_{\R^p}[[\xi]]$ are \emph{formal power series} in the $\Z_2^n$-commutative parameters $\xi$, with coefficients in smooth functions:
$$
 \Ci_{\R^p}(U)[[\xi]]:=\left\{ \sum_{\alpha\in\N^{\times |\mathbf{q}|}} f_{\alpha}(x)\,\xi^{\alpha}\; |\; f_{\alpha}\in\Ci(U)\right\}\quad(U\,\text{open in}\;\,\R^p) \;.
$$

\emph{$\Z_2^n$-morphisms} between $\Z_2^n$-manifolds are just morphisms of $\Z_2^n$-ringed spaces, i.e., pairs
$\zF=(\phi,\phi^*):(M,\cO_M)\to (N,\cO_N)$ made of a continuous map $\phi:M\to N$ and a sheaf morphism $\phi^*:\cO_N\to\phi_*\cO_M$, i.e., a family of $\Z_2^n$-graded unital $\R$-algebra morphisms, which commute with restrictions and are defined, for any open $V\subset N$, by $$\phi^*_V:\cO_N(V)\to \cO_M(\phi^{-1}(V))\;.$$

We denote the category of $\Z_2^n$-manifolds and $\Z_2^n$-morphisms between them by $\Z_2^n{\tt Man}$.
\end{defi}

\begin{rem}
When considering sheaves, like, e.g., $\cO_M$, we sometimes omit the underlying topological space $M$, provided this space is clear from the context.\end{rem}

\begin{rem} Let us stress that the base space $M$ corresponds to the degree zero coordinates (and not to the even degree coordinates), and let us mention that it can be proven that the topological base space $M$ carries a natural smooth manifold structure of dimension $p$, that the continuous base map $\phi:M\to N$ is in fact smooth, and that the algebra morphisms $$\phi^*_m:\cO_{\phi(m)}\to\cO_{m}\quad(m\in M)\;$$ between stalks, which are induced by the considered $\Z_2^n$-morphism $\Phi:{\cal M}\to {\cal N}$, respect the unique homogeneous maximal ideals of the local graded rings $\cO_{\phi(m)}$ and $\cO_{m}$.
\end{rem}

\section{Linear $\Z_2^n$-algebra}\label{LinAlg}

Let us mention that $\Z_2^n$-graded modules $V$ over a $\Z_2^n$-commutative algebra $A$ are defined canonically, $$A^{\zg'}\cdot V^{\zg''}\subset V^{\zg'+\zg''}\quad(\zg',\zg''\in\Z_2^n)\;,$$ and that $\Z_2^n$-graded modules over a fixed $\Z_2^n$-commutative algebra, and degree-preserving $A$-linear maps between them form an abelian category $\Z_2^n{\tt Mod}(A)$. This category admits a natural \emph{symmetric monoidal structure} $\0_A$, with braiding given by
$$
\begin{array}{rccc}
c_{\op{VW}}^{\op{gr}}: &V \otimes_A W &\to& W \otimes_A V\\
&v \otimes w &\mapsto& (-1)^{\langle \deg v,\deg w\rangle} w \otimes v\;,
\end{array}
$$
for homogeneous elements $v$ and $w$. This structure is also \emph{closed}, as for every $V\in\Z_2^n{\tt Mod}(A)$, the functor $$-\0_A V: \Z_2^n{\tt Mod}(A) \to \Z_2^n{\tt Mod}(A)$$ has a right-adjoint $$\underline{\op{Hom}}_A(V,-): \Z_2^n{\tt Mod}(A) \to \Z_2^n{\tt Mod}(A)\;,$$ i.e., for any $U,W\in\Z_2^n{\tt Mod}(A)$, there is a natural isomorphism
$$
\op{Hom}_A(U\0_A V, W)\simeq \op{Hom}_A(U,\underline{\op{Hom}}_A(V,W))\;,
$$
where $\Hom_A(V,W)$ denotes the {\it categorical hom} made of the degree-respecting $A$-linear maps, i.e., of the $A$-linear maps $\ell:V\to W$ of degree $0\in\Z_2^n$, i.e., $\ell(V^\zg)\subset W^{\zg+0},$ for all $\zg\in\Z_2^n$. One can readily verify that the \emph{internal hom} $\underline{\Hom}_A(V,W)$ is the $\Z_2^n$-graded $A$-module which consists of all $A$-linear maps $\ell: V \to W$ of all possible degrees $\za\in\Z_2^n$, i.e.,
$$
\ell(V^{\gamma})\subset W^{\gamma+\za}\;,
$$
for all $\gamma\in\Z_2^n$. The $A$-linear maps of degree $\za$ constitute the $\za$-part $\underline{\Hom}^{\za}_A(V,W)$ of $\underline{\Hom}_A(V,W)$. Hence, contrary to the case of modules over a classical commutative algebra, the internal hom $\underline{\Hom}_A$ differs from the categorical hom $\Hom_A$, since this latter contains only $0$-degree $A$-linear maps: $\Hom_A(V,W)=\underline{\Hom}^{0}_A(V,W)$.

\section{Sheaves of differential operators on a $\Z_2^n$-manifold}

It is known that bump functions and partitions of unity do exist in $\Z_2^n$-manifolds \cite{CGP}, and that they can be used similarly to classical bump functions in smooth geometry \cite{Lei}. Recall that in a $\Z_2^n$-manifold ${\cal M}=(M,\cO_M)$, the support of a $\Z_2^n$-function $s\in\cO_M(U)$, $U\subset M$, is defined as usual as the complementary in $U$ of the open subset of identical zeros of $s$ in $U$. Further, a bump function $\zg$ of $\cal M$ around $m\in M$ is a globally defined degree zero $\Z_2^n$-function $\zg\in\cO^0_M(M)$ for which there exist open neighborhoods $W\subset V\subset M$ of $m$, such that $\op{supp}\zg\subset V$ and the restriction $\zg|_W=1$. Moreover, if $\ze$ is the projection to the base that implements the short exact sequence of sheaves \cite{CGP} \be\label{FundaSES}0\to\ker_M\ze\to \cO_M\stackrel{\ze}{\to}\Ci_M\to 0\;,\ee the projection $\ze_M(\zg)\in\Ci_M(M)$ is a bump function of $M$ around $m$.\medskip

Consider now a $\Z_2^n$-manifold ${\cal M}=(M,\cO_M)$. Notice that $A=\R$ can be viewed as a $\Z_2^n$-commutative algebra concentrated in degree 0, and that, for any open $U\subset M$, the algebra $\cO(U)$ is an object in $\Z_2^n{\tt Mod}(\R)$, i.e., is a $\Z_2^n$-graded $\R$-vector-space. Hence, according to Section \ref{LinAlg}, $$\underline{\operatorname{End}}_\R(\cO(U)):=\underline{\Hom}_\R(\cO(U),\cO(U))\;$$ is the $\Z_2^n$-graded $\R$-vector-space of all $\R$-linear maps of all $\Z_2^n$-degrees from $\cO(U)$ to itself. Composition endows this space with a $\Z_2^n$-graded associative unital $\R$-algebra structure, and the $\Z_2^n$-graded commutator $[-,-]$ endows this space with a $\Z_2^n$-graded Lie algebra structure. We can identify $\cO(U)$ with an associative subalgebra of $\underline{\operatorname{End}}_\R(\cO(U))$ using the left-regular representation $$g\mapsto m_g\,,\;\, m_g(f)=g\cdot f\;.$$
\begin{defi}\label{AlgDefDO} The $\Z_2^n$-graded $\cO(U)$-module of {\it $k$-th order differential operators} $\cD^k(U)$, $k\in\N$, is defined inductively by
\be\label{Vin}\cD^k(U):=\{D\in \underline{\operatorname{End}}_\R(\cO(U)):[D,\cO(U)]\subset\cD^{k-1}(U)\},\ee where $\cD^{-1}(U)=\{0\}.$\end{defi}

Of course $\cD^0(U)=\cO(U)$. Indeed, it is clear that $\cO(U)\subset\cD^0(U)$. Conversely, if $D\in\cD^0(U)$, then, for any $f\in\cO(U)$, we have $[D,f]=0$, so that $$D(f)=D(f\cdot 1)=D(1)\cdot f=m_{D(1)}(f)\;,$$ i.e., $D\simeq D(1)\in\cO(U)$. It follows that 0-order differential operators are local, i.e., if $D\in\cD^0(U)$, $f\in\cO(U)$, and $-|_V$ denotes the restriction to an open subset $V\subset U$, then $D(f)|_V$ only depend on $f|_V$, or, equivalently, if $f|_V=0$ then $D(f)|_V=0$. This implies by induction that any differential operator of any order is local. Indeed, if $D\in\cD^k(U)$, if $f\in\cO(U)$ and $f|_V=0$, and if $m\in V$, let $\zg\in \cO^0(M)$ be a bump function of $\cal M$ around $m$ with support $\operatorname{supp}\zg\subset V$ and restriction $\zg|_W=1$, for some neighborhood $W\subset V$ of $m$. It then follows from the defining property of differential operators applied to $[D,\zg]f$, the induction assumption, and the fact $\zg f=0$, that $(Df)|_W=0$.\medskip

We now show that any $D\in\cD^k(U)$ can be localized.

\begin{prop}\label{Localization} Let $U\subset M$ and $V\subset U$ be open. Any $D\in\cD^k(U)$ induces a unique $D|_V\in\cD^k(V)$. This restricted operator satisfies $D|_V(g|_V)=(Dg)|_V$, for any $g\in\cO(U)$.\end{prop}

\begin{proof} If $f\in\cO(V)$ and $m\in V$, we choose a function $F\in\cO(U)$ such that $F|_W=f|_W$, for some neighborhood $W\subset V$ of $m$ (it suffices to choose a bump function $\zg$ in $\cM$ around $m$ with support in $V$ and restriction 1 to $W$ ($m\in W\subset V$) and to set $F=\zg f$). Locality of $D$ implies that the section $(DF)|_W\in\cO(W)$ and the section $(DF')|_{W'}\in\cO(W')$ obtained similarly but for a point $m'\in V$, depend only on $f$ and thus coincide on the intersection $W\cap W'$. Hence, these local sections define a unique global section $D|_Vf\in\cO(V)$ such that $(D|_Vf)|_W=(DF)|_W.$ In particular, if $f=g|_V$ is the restriction of a function $g\in\cO(U)$, we can take $F=g$, so that \be\label{DefRest}D|_V(g|_V)=(Dg)|_V\;,\ee as announced. Since, obviously, $D|_V\in\underline{\operatorname{End}}_\R(\cO(V))$ (note that $D|_V$ has the same parity as $D$), it suffices -- to prove Proposition \ref{Localization} -- to observe that, for any $f_1,\ldots,f_{k+1},g\in\cO(V)$, we have $$([\ldots[[D|_V,f_1],f_2],\ldots,f_{k+1}]g)|_W=([\ldots[[D,F_1],F_2],\ldots,F_{k+1}]G)|_W=0,$$ with obvious notation.\end{proof}

The assignment $\cD^k:U\to \cD^k(U)$ of the $\Z_2^n$-graded $\cO(U)$-module $\cD^k(U)$ to any open subset $U\subset M$ is a presheaf of $\Z_2^n$-graded $\cO$-modules. Indeed, the restriction maps $\zr^U_V:\cD^k(U)\ni D\mapsto D|_V\in\cD^k(V)$ are clearly of degree 0 and $\cO$-linear (i.e., $\zr^U_V(GD)=r^U_V(G)\zr^U_V(D)$, where $r^U_V$ is the restriction map of the function sheaf $\cO$), and they satisfy the usual compatibility conditions for restriction maps. The presheaf $\cD^k$ is in fact a sheaf over $M$ -- the standard proof goes through. If we wish to emphasize the base topological space of this sheaf we write $\cD^k_M$ instead of $\cD^k$.\medskip
\newcommand{\Zn}{\Z_2^n}

Some continuity results will be needed. Let $\cM=(M,\cO_M)$ be as above a $\Z_2^n$-manifold. The kernel sheaf $\ker_M\ze$ of the base projection sheaf morphism $\ze:\cO_M\to\Ci_M$ (see \eqref{FundaSES}) is a sheaf of ideals of the structure sheaf $\cO_M$, which we usually denote by $\cJ_M=\ker_M\ze$. Further, we write $\fm$ for the unique homogeneous maximal ideal of the stalk $\cO_m$ of the sheaf $\cO_M$ at $m\in M$. The topology of $\cO_M$ associated to the filtration $$\cO_M\supset \cJ_M\supset \cJ_M^2\supset\ldots$$ is referred to as the $\cJ_M$-adic topology of the structure sheaf. Similarly, the topology of $\cO_m$ associated to the filtration $$\cO_m\supset\fm\supset\fm^2\supset\ldots$$ is the $\fm$-adic topology of the stalk. The $\Z_2^n$-function sheaf $\cO_M$ (resp., the $\Z_2^n$-function algebra $\cO_M(U)$, $U$ open in $M$) is Hausdorff-complete with respect to the $\cJ_M$-adic (resp., $\cJ_M(U)$-adic) topology \cite{CGP}. Moreover, $\Z_2^n$-morphisms are $\cJ$-adically continuous and the induced morphisms between stalks are $\frak m$-adically continuous \cite{CGP}. Moreover:

\begin{prop}\label{ContDO}
Any differential operator $\zD\in\cD^k(U)$ of any degree $k\ge 0$ and over any open $U\subset M$ is $\cJ(U)$-adically continuous. In particular, if a sequence of functions $f_n\in\cO(U)$ tends $\cJ(U)$-adically to a function $f\in\cO(U)$, then $\zD f_n$ tends $\cJ(U)$-adically to $\zD f$. Due to its locality, any differential operator $\zD$ of order $k$ acting on $\cO(U)$ canonically induces, for any $m\in M$, a differential operator $\zD_m$ of order $k$ acting on $\cO_m$, and this induced operator $\zD_m$ is $\fm$-adically continuous.
\end{prop}

The proof is based on the following lemmata. We will omit the subscript $M$ which remembers that most of the sheaves considered here are defined over $M$.

\begin{lem}\label{DOandJadictop} Let $\cM=(M,\cO_M)$ be a $\Z_2^n$-manifold, let $k\ge 0$ be an integer and $U\subset M$ be open, and let $\zD^k\in\cD^k(U)$. We have $$\zD^k\cJ^{k+c}(U)\subset \cJ^c(U)\;,$$ for all $c\ge 1$.
\end{lem}

\begin{proof}
In this proof we just write $\zD$ thus omitting the superscript indicating its order. If $k=0$, the differential operator $\zD$ is a function $\zD\in\cO(U)$. Hence, for any $c\ge 1$, $$\zD\, \cJ^{c}(U)=\zD\cdot \cJ^c(U)\subset \cJ^c(U)\;.$$ Assume now that the claim holds for $k\ge 0$. We will show that it is then also valid for $k+1$. Let $c=1$, let $\zD\in\cD^{k+1}(U)$, let $f\in \cJ(U)$ and let $g\in \cJ^{k+1}(U)$. We have $f\zD g\in \cJ(U)$, since $\cJ(U)$ is an ideal, and, since $[\zD,f]\in\cD^k(U)$, we have $[\zD,f]g=\zD(fg)\pm f\zD g\in \cJ(U),$ in view of the induction assumption. It follows that $\zD(fg)\in \cJ(U)$ and that $\zD\, \cJ^{k+2}(U)\subset \cJ(U)$: for differential operators of order $k+1$ the claim is true for the lowest value $c=1$. It thus suffices to assume that the statement holds for $c\ge 1$ and to prove that it holds then as well for $c+1$. Therefore, consider $f\in \cJ(U)$ and $g\in \cJ^{k+c+1}(U)$. Since $[\zD,f]g=\zD(fg)\pm f\zD g\in \cJ^{c+1}(U)$, due to the induction assumption on $k$, and since $f\zD g\in \cJ^{c+1}(U)$, due to the induction assumption on $c$, we finally obtain that $\zD\, \cJ^{k+c+2}(U)\subset \cJ^{c+1}(U)$, what completes the proof.
\end{proof}

The next lemma is similar and has the same proof.

\begin{lem}\label{DOandmadictop} Let $\cM=(M,\cO_M)$ be a $\Z_2^n$-manifold, let $k\ge 0$ be an integer and $U\subset M$ be open, and let $\zD^k\in\cD^k(U)$. If $m\in U$ and $\frak{m}_m$ is the unique homogeneous maximal ideal of $\cO_m$, we have
$$\zD_m^k\frak{m}^{k+c}_m\subset \frak{m}^c_m\;,$$ for all $c\ge 1$.
\end{lem}

We are now prepared to prove Proposition \ref{ContDO}.

\begin{proof} Let $\zD\in\cD^k(U)$ ($k\ge 0$, $U\subset M$ open), let $g \in \cO(U)$ and $c\ge 1$, and show that the preimage $\zD^{-1}(g+\cJ^c(U))$ of the arbitrary basic open subset $g+\cJ^c(U)$, is open in the $\cJ(U)$-adic topology. The fact that $\zD\, \cJ^{k+c}(U) \subset \cJ^c(U)$ implies that $\zD^{-1}(g + \cJ^c(U))$ is the union of the $\cJ(U)$-adically open sets $f + \cJ^{k+c}(U)$, where $f$ runs over all elements of $\zD^{-1}(g + \cJ^c(U))$, so that $\zD^{-1}(g + \cJ^c(U))$ is actually open. The claim that the limit of $\zD f_n$ is the value of $\zD$ at the limit of $f_n\,$, is a direct consequence of Lemma \ref{DOandJadictop}. The statement concerning $\zD_m$ follows analogously from Lemma \ref{DOandmadictop}.\end{proof}

In the next proposition, we consider $p|\mathbf{q}$ local coordinates $u=(x,\xi)$ and set $u=(x,\zy,\zz)$, where $x$ (resp., $\zy$, $\zz$) are the coordinates of degree zero (resp., of even non-zero degrees, of odd degrees). More precisely, as suggested above, we fix the coordinate order: \be\label{CoordOrder}u=(x^1,\ldots,x^p,\zy^1,\ldots,\zy^{q_1},\zy^{q_1+1},\ldots,\zy^{\sum_{i=1}^{2^{n-1}-1}q_i},\zz^1,\ldots,\zz^{q_{2^{n-1}}},\zz^{q_{2^{n-1}}+1},\ldots,\zz^{\sum_{i=2^{n-1}}^{2^n-1}q_i})\;.\ee The corresponding $\Z_2^n$-graded derivations $\p_\zz,\p_\zy,\p_x$ are written in decreasing order.

\begin{theo}\label{LocBasisDO} For any $k\in\N$, the sheaf $\cD_M^k$ of $k$-th order differential operators of a $\Z_2^n$-manifold ${\cal M}=(M,\cO_M)$ of dimension $p|\mathbf{q}$ is a locally free sheaf of $\Z_2^n$-graded $\cO_M$-modules, with local basis
$$\p^\zg_{\zz}\partial^{\zb}_{\zy}\p_x^\za\;,$$ where $(x,\zy,\zz)$ are local coordinates, $\zg^a\in\Z_2$ and $\zb^b,\za^c\in\N$, and $|\za|+|\zb|+|\zg|\le k.$
\end{theo}

\begin{rem} Note that the decomposition of a differential operator of order $k$ in this local basis leads to a finite sum $(\,$see Equation \eqref{LocForm} below$\,)$.\end{rem}

\begin{proof} Since we work locally, we take ${\cal M}=(U,\Ci_U[[\zy,\zz]])$, where $U$ is an open subset of $\R^p$. We first prove uniqueness of the coefficients. If $D\in\cD^k(U)$ {\it is} of the type
\be \sum_{i=0}^kD^i=\sum_{i=0}^k\sum_{|\za|+|\zb|+|\zg|=i}D^i_{\za\zb\zg}(x,\zy,\zz)\;\p_{\zz}^{\zg}\,\p_\zy^\zb\,\p_x^\za\;\in\cD^k(U),\label{LocForm}\ee and if $m_{\za\zb\zg}=\frac{1}{\za!\zb!}x^{\za}\zy^{\zb}\zz^\zg,$ where the coordinates are written in increasing order and $|\za|+|\zb|+|\zg|=i$, then necessarily \be D^i_{\za\zb\zg}=D^im_{\za\zb\zg}=Dm_{\za\zb\zg}-\sum_{j=0}^{i-1}D^jm_{\za\zb\zg},\label{CoeffLocForm}\;.\ee Indeed, the term with same indices as the considered monomial reduces to its coefficient when applied to the monomial, and any other term contains at least one index that is higher than the corresponding index in the monomial, so that this term annihilates the monomial). Hence, the coefficients $D^i_{\za\zb\zg}$ of $D$ are unique, if they exist. More precisely, for $|\zm|+|\zn|+|\zp|=1$ and $|\za|+|\zb|+|\zg|=2$, we get $$D^0_{000}=Dm_{000}\;,$$ $$D^1_{\zm\zn\zp}=Dm_{\zm\zn\zp}-Dm_{000}\cdot m_{\zm\zn\zp}\;,$$ \be\label{Explicit}D^2_{\za\zb\zg}=Dm_{\za\zb\zg}-\sum_{|\zm|+|\zn|+|\zp|=1}(Dm_{\zm\zn\zp}-Dm_{000}\cdot m_{\zm\zn\zp})\;\p_\zz^\zp\,\p_\zy^\zn\,\p_x^\zm \; m_{\za\zb\zg}-Dm_{000}\cdot m_{\za\zb\zg}\;\ldots\ee

Consider now an arbitrary $D\in\cD^k(U)$ and set $\zD=D-\zS\in\cD^k(U)$, where $\zS$ denotes the {\small RHS} of (\ref{LocForm}) with the coefficients defined in (\ref{CoeffLocForm}). This operator $\zD$ vanishes by construction
on the polynomials of degree $\le k$ in $x,\zy,\zz$. The reader might wish to check this claim by direct computation. For instance, the computation for $k=2$ is straightforward in view of Equation \eqref{Explicit}.\medskip

Note now that, for any $f_1,\ldots,f_{\ell-1},h\in\Ci_U(U)[[\zy,\zz]]$, $\ell\ge k+1$, we have
\be \zD(f_1\ldots f_{\ell-1}h)=\sum_{b=1}^{\ell-1}\sum\pm
f_{i_1}\ldots f_{i_b}\zD(f_{i_{b+1}}\ldots f_{i_{\ell-1}}h)+
F(h),\label{InducVanishPoly}\ee
as immediately seen when developing $F(h):=[\ldots[[\zD,f_1],f_2],\ldots,f_{\ell-1}]h.$ If $\ell>k+1$, the term $F(h)$ vanishes, whereas in the case $\ell=k+1$ it is given by $F(h)=F(1)h$. Equation (\ref{InducVanishPoly}) shows that $\zD=0$ on any polynomial of degree $k+1$ in $x,\zy,\zz$. Indeed, for coordinate functions $f_1,\ldots,f_{k},h$, the value $\zD(f_1\ldots f_{k}h)$ vanishes if and only if $F(1)$ vanishes. However, the value $F(1)$ is, again in view of (\ref{InducVanishPoly}), computed from values of $\zD$ on polynomials of degree $\le k$ and does therefore vanish. It now follows by induction from (\ref{InducVanishPoly}) that $\zD=0$ on an arbitrary polynomial in $x,\zy,\zz$. We can extend this conclusion to polynomial formal series, i.e., to series $$\mathcal{P} = \sum_{|\mu| \geq 0} P_\mu(x) \xi^\mu\;,$$ where $\mu$ is a multi-index with sum of components denoted by $|\zm|$, and where $P_\mu(x)$ is a polynomial in the zero-degree coordinates $x$. Indeed, the sequence $\mathcal{P}_c$ obtained by taking in the series $\mathcal{P}$ only the terms $|\mu|< c$ converges $\cJ(U)$-adically to the series $\cP$, since in local coordinates the ideal $\cJ^c(U)$ is made of the series $$\sum_{|\zm|\ge c}f_\zm(x)\xi^\zm\;,$$ all whose terms contain at least $c$ formal coordinates. Hence, Proposition \ref{ContDO} implies that $\zD\cP\in\cJ^c(U)$, for all $c$, i.e., that $\zD\cP=0$. Let now $f \in \Ci_U(U)[[\zy,\zz]]$ and let $m$ be any point in $U$. By the polynomial approximation theorem \cite[Theorem 6.10]{CGP}, there exists, for any $c$, a polynomial formal series $\mathcal{P}$ such that $[f]_m - [\mathcal{P}]_m \in \mathfrak{m}^{k+c}_m$ ($\cP$ depends on $f,m,$ and $c$). When applying the differential operator $\zD_m$ of order $k$ to $[f]_m - [\mathcal{P}]_m$, we get from Lemma \ref{DOandmadictop} that $[\zD f]_m\in\mathfrak{m}^{c}_m$, for any $c$, so that $[\zD f]_m = 0$. Since $m \in U$ was arbitrary, $\zD f = 0$ for any function $f\in\Ci_U(U)[[\zy,\zz]]$, i.e., $D = \zS$, where $\zS$ is, as above, the {\small RHS} of Equation \eqref{LocForm}.
\end{proof}

The definition of $\cD^1(U)$ implies straightforwardly that $$\cD^1(U)=\cO(U)\oplus{\op{Der}}_\R(\cO(U))\;,$$ where the second term of the {\small RHS} is the $\Z_2^n$-graded $\cO(U)$-module of vector fields over $U$. Further, the $\R$-vector space $\underline{\operatorname{End}}_\R(\cO(U))$ carries natural $\Z_2^n$-graded associative and $\Z_2^n$-graded Lie algebra structures $\circ$ and $[-,-]$, where $[-,-]$ is the above-introduced $\Z_2^n$-graded commutator. An induction on $k+\ell$ allows seeing that $\cD^k(U)\circ\cD^{\ell}(U)\subset\cD^{k+\ell}(U)$ and $[\cD^k(U),\cD^{\ell}(U)]\subset\cD^{k+\ell-1}(U)$, so that the (colimit or direct limit) vector space $\cD(U):=\cup_{k\in\N}\cD^k(U)$ of all differential operators inherits $\Z_2^n$-graded associative and $\Z_2^n$-graded Lie algebra structures that have weight $0$ and $-1$, respectively, with respect to the filtration degree. The assignment $\cD:U\to\cD(U)$ is a locally free sheaf of $\Z_2^n$-graded $\cO$-modules and of $\Z_2^n$-graded associative and Lie algebras over $M$. The algebra $\cD(M)$ is the $\Z_2^n$-graded Lie algebra of {\it differential operators} of the $\Z_2^n$-manifold $\cM$.

\begin{rem} The reader observed probably that (as usual) $\cD^{k-1}(U)\subset\cD^{k}(U)$. Hence, a $k$-th order differential operator is actually a differential operator of order $\le k$. To emphasize this fact some authors write $\cD^{\le k}(U)$ instead of $\cD^k(U)$.\end{rem}

\section{Functional analytic properties of the function sheaf of a $\Z_2^n$-manifold}

For a review of Fr\'echet spaces, algebras, and sheaves, we refer the reader to the Appendix.\medskip

We define $\Z_2^n$-graded Fr\'echet vector spaces, nuclear vector spaces, Fr\'echet algebras and Fr\'echet sheaves.

\begin{defi}\label{FNAS} \begin{itemize}\ \item A \emph{$\Z_2^n$-graded Fr\'echet vector space} is a $\Z_2^n$-graded vector space $V$, all whose homogeneous subspaces $V^\zg$, $\zg\in\Z_2^n$, are Fr\'echet vector spaces. We denote by $(p_i^\zg)_{i\in I}$ a family of seminorms corresponding to $V^\zg$. \item A \emph{$\Z_2^n$-graded nuclear {\small LCTVS}} is a $\Z_2^n$-graded vector space $V$, all whose homogeneous subspaces $V^\zg$, $\zg\in\Z_2^n$, are nuclear. \item A \emph{$\Z_2^n$-graded $\Zn$-commutative (nuclear) Fr\'echet algebra} is a $\Z_2^n$-graded (nuclear) Fr\'echet vector space $A$ that is equipped with a $\Z_2^n$-graded $\Zn$-commutative associative bilinear multiplication $\cdot:A^{\zg'}\times A^{\zg''}\to A^{\zg'+\zg''}$, such that there are equivalent countable families of seminorms that are submultiplicative, i.e., equivalent countable families $(q_n^\zg)_{n\in \N}$ of seminorms, such that, for all $n\in \N$, $$q_n^{\zg'+\zg''}(x\cdot y)\le q_n^{\zg'}(x)\,q_n^{\zg''}(y),\;\forall x\in A^{\zg'}, \forall y\in A^{\zg''}\;.$$ \item A \emph{(nuclear) Fr\'echet sheaf of $\Z_2^n$-graded $\Zn$-commutative algebras} is a sheaf $\cF$ of $\,\Z_2^n$-graded $\Zn$-commutative (real) algebras over a smooth manifold $M$, such that all section spaces $\cF(U)$ are $\Z_2^n$-graded $\Zn$-commutative (nuclear) Fr\'echet algebras and the locally convex topology on $\cF(U)$ is the coarsest topology for which all restriction maps $\cF(U)\to\cF(U_i)$ are continuous.\end{itemize}\end{defi}

The algebraic direct sum $\oplus_\za V_\za$ of a family $(V_\za)_{\za\in \cA}$ of {\small LCTVS}-s $V_\za$ is usually equipped with the {\it direct sum topology}, that is, with the finest locally convex topology such that the injections $i_\za:V_\za\to \oplus_\za V_\za$ are all continuous. In this case, we refer to the direct sum as the {\it topological direct sum of the {\small LCTVS}-s} $V_\za$. It is known that a countable topological direct sum of nuclear {\small LCTVS}-s is a nuclear {\small LCTVS}. Further, a finite topological direct sum of Fr\'echet spaces is a Fr\'echet space. These results show that a $\Z_2^n$-graded Fr\'echet vector space (resp., a $\Z_2^n$-graded nuclear {\small LCTVS}) is a Fr\'echet space (resp., a nuclear {\small LCTVS}), when equipped with the direct sum topology.

\begin{rem} In the following, we suppress the superscript $\zg$ in the various seminorms $p_i^\zg$, $q_n^\zg$, ... that we consider. In other words, $p_i$, $q_n$, ... refer to a seminorm of some space $V^\zg$.\end{rem}

We are now prepared to prove one of the main theorems of this paper.

\begin{theo}\label{MainTheo} The function sheaf $\cO_M$ of a $\Z_2^n$-manifold $\cM=(M,\cO_M)$ is a nuclear Fr\'echet sheaf of $\Z_2^n$-graded $\Zn$-commutative algebras.\end{theo}

\begin{proof} Let $U\in{\tt Open}(M)$, let $C$ be any compact subset of $U$, and let $D$ be any differential operator in $\cD(U)$. For any $f\in\cO(U)$, we set \be\label{StandSemiNormTop}p_{C,D}(f)=\sup_{x\in C}|\ze(D(f))(x)|\;,\ee where $\ze$ is the projection $\ze:\cO_M\to \Ci_M$, see Equation \eqref{FundaSES}. It is obvious that each $p_{C,D}$ is a seminorm on $\cO(U)$.

\begin{lem}\label{LEM1} For any $\,U\in{\tt Open}(M)$, the family of seminorms $(p_{C,D})_{C,D}$ on $\cO(U)$ is separating and endows $\cO(U)$ with a Hausdorff locally convex topological vector space structure. \end{lem}

\begin{proof} It suffices to prove the separability. If $\sup_{x\in C}|\ze(D(f))(x)|=0$, for all $C$ and all $D$, then $\ze(D(f))=0$ in $U$ for any $D$, since $U$ admits a (countable) cover by compact subsets $C$, see Lemma \ref{CompCov}. Differently stated, we have \be\label{U}D(f)\in \ker\ze_U=\cJ(U)\;,\ee for any $D\in\cD^k(U)$ and for any $k\in\N$. Let now $(V_i)_{i\in\N}$ be a (countable) cover of $U$ by $\Z_2^n$-chart domains (any open cover of $U$ admits a countable subcover, see proof of Lemma \ref{CompCov}) and let $V$ be any element of this cover. For any $m\in V$, there is a $\Z_2^n$-bump-function $\zg\in\cO^0(U)$ and neighborhoods $N_1\subset N_2\subset V$ of $m$ such that $\zg|_{N_1}=1$ and $\op{supp}\zg\subset N_2\,$. In view of Equation (\ref{U}), we have \be\label{N}D|_{N_1}(f|_{N_1})=D(f)|_{N_1}\in\cJ(N_1)\;,\ee for any $D\in\cD^k(U)$ and any $k\in\N$. If we choose $D=1\in\cO(U)=\cD^0(U)$ in (\ref{N}), we can conclude that $f|_{N_1}\in\cJ(N_1)$. It turns out that, if $f|_{N_1}\in \cJ^{k-1}(N_1)$, $k\ge 2$, then $f|_{N_1}\in\cJ^k(N_1)\,$. Indeed, denote the coordinates in $V$ by $u=(x,\xi)$ and write $$f|_{N_1}(x,\xi)=\sum_{\ell=k-1}^\infty\sum_{|\zb|=\ell}f_{\zb}(x)\,\xi^\zb\;.$$ Our goal is to show that $f|_{N_1}\in \cJ^k(N_1)$, i.e., that all coefficients $f_\zb(x)$, $|\zb|=k-1$, vanish in $N_1$. Let ${\frak B}$ be any multiindex such that $|{\frak B}|=k-1$. When using the operator $\zg\,\p_\xi^{\frak B}\in\cD^{k-1}(U)$, we get from (\ref{N}), $$\cJ(N_1)\ni\p_\xi^{\frak B}\sum_{\ell=k-1}^\infty\sum_{|\zb|=\ell}f_{\zb}(x)\,\xi^\zb=\sum_{\ell=k-1}^\infty\sum_{|\zb|=\ell}f_{\zb}(x)\,\p_\xi^{\frak B}\,\xi^\zb\;,$$ as $\p_\xi^{\frak B}$ is $\cJ(N_1)$-adically continuous. Since $\p_\xi^{\frak B}\xi^{\frak B}={\frak B}!\,$, it follows that $f_{\frak B}$ vanishes in $N_1$. Hence, $f|_{N_1}\in\cJ^k(N_1)$, so $f|_{N_1}\in\cJ^r(N_1)$ for all $r$, and thus $f|_{N_1}=0$. Finally, we get $f|_V=0$ and $f=0\,$. \end{proof}

The next lemma covers the case where $U$ is a $\Z_2^n$-chart domain.

\begin{lem}\label{LEM2} Let $U\subset M$ be a $\Z_2^n$-chart domain with coordinates $u=(x,\xi)$.
\begin{itemize} \item Let $$f_n=\sum_\zb f_{n\zb}(x)\xi^\zb\quad(\text{resp.,}\;f=\sum_\zb f_\zb(x)\xi^\zb)$$ be a sequence of functions (resp., a function) in $\cO(U)$. The sequence $f_n$ is Cauchy in $\cO(U)$ $(\,$resp., converges to $f$ in $\cO(U)$$\,)$ if and only if the sequences $f_{n\zb}$ are all Cauchy in $\Ci(U)$ $(\,$resp., converge all to the corresponding $f_\zb$ in $\Ci(U)$$\,)$. \item The locally convex Hausdorff space $\cO(U)$ is complete. \item The space $\cO(U)$ is a $\Z_2^n$-graded nuclear Fr\'echet algebra.
\end{itemize}
\end{lem}

We start with the following observation. Let $(U,u=(x,\xi))$ be a $\Z_2^n$-coordinate system, $D\in\cD^k(U)$, and $f\in\cO(U)$. We have $$\ze(D(f))=\ze\sum_{\ell=0}^k\sum_{|\za|+|\zb|=\ell} D_{\za\zb}(x,\xi)\, \p_\xi^\zb \p_x^\za\,\sum_{\zg}f_\zg(x)\xi^\zg=$$ $$\sum_{\ell=0}^k\sum_{|\za|+|\zb|=\ell}\,\ze(D_{\za\zb}(x,\xi))\;\ze\sum_{\zg}\p_x^\za f_\zg(x)\,\p_\xi^\zb\xi^\zg\;.$$ If $\zb\neq\zg$, either there is $\zb_i>\zg_i\,$, or all $\zb_i\le\zg_i$ but for at least one $i$ we have $\zb_i<\zg_i$. In the first case $\p_\xi^\zb\xi^\zg=0$ and in the second $\p_\xi^\zb\xi^\zg\in\cJ(U)$, so that in both situations the corresponding term in the series over $\zg$ vanishes under the action of $\ze$. If $\zb=\zg$, the derivative with respect to $\xi$ equals $\zb!\,$, so that \be\label{EpsDiff}\ze(D(f))=\sum_{\ell=0}^k\sum_{|\za|+|\zb|=\ell}\,\ze(D_{\za\zb}(x,\xi))\;\zb!\,\p_x^\za f_\zb(x)\;.\ee

We are now prepared for the proof of Lemma \ref{LEM2}.

\begin{proof} $\bullet\,$ Assume first that the $f_{n\zb}$ are all Cauchy in the locally convex topology of $\Ci(U)$: for any base differential operator $\zD$ (acting on base functions $\Ci(U))$ and any compact $C\subset U$, we have \be\label{CyCi}p_{\zD,C}(f_{r\zb}-f_{s\zb})=\sup_C|\zD(f_{r\zb}-f_{s\zb})|\to 0\;,\ee if $r,s\to\infty$, see Example \ref{FrechetSmoothFctBase}. In this case, we get, for any $\Z_2^n$-differential operator $D$ (acting on $\Z_2^n$-functions $\cO(U)$) and any compact $C\subset U$, $$p_{C,D}(f_r-f_s)=\sup_C|\ze(D(f_r-f_s))|\le\sum_{\za\zb}\,\zb!\,\sup_C|\ze(D_{\za\zb}(x,\xi))|\,\sup_C|\p_x^\za (f_{r\zb}-f_{s\zb})|\to 0\;,$$ if $r,s\to\infty$, so that $f_n$ is Cauchy in the topology of $\cO(U)$. Conversely, if $f_n$ is Cauchy in $\cO(U)$, we have to show that \eqref{CyCi} holds for any $\zD$, $C$, and $\zb$. Fix these three data. The base differential operator $\zD$ reads $$\zD=\sum_\za \ \zD_\za(x)\,\p_x^\za$$ and $D=\p_\xi^\zb\zD$ is a $\Z_2^n$-differential operator. In view of \eqref{EpsDiff}, we get $$\sup_C|\ze(D(f_r-f_s))|=\zb!\,\sup_C|\sum_\za \zD_\za(x)\,\p_x^\za(f_{r\zb}-f_{s\zb})|=\zb!\,\sup_C|\zD(f_{r\zb}-f_{s\zb})|\;.$$ The proof of the convergence statement of item one in Lemma \ref{LEM2} is similar.\medskip

$\bullet\,$ In view of Example \ref{FrechetSmoothFctBase} and item one, item two is obvious.\medskip

$\bullet\,$ To prove that the complete Hausdorff locally convex topological vector space $\cO(U)$ is a Fr\'echet space, it suffices to show that there exists a countable family of seminorms on $\cO(U)$ that is equivalent to the family $(p_{C,D})_{C,D}$. To prove that $\Ci(U)$ is a Fr\'echet space, one uses a countable cover of $U$ by compact subsets $C_n\subset U$ such that $C_n$ is contained in the interior of $C_{n+1}$ \cite{Tre}. Proceeding similarly, we define the family \be\label{Cnalphabeta}p_{C_n,\za,\zb}(f)=\sup_{C_n}|\ze(\p_\xi^\zb\p_x^\za f)|\;,\ee where the components of $\za$ belong to $\N$ and the components of $\zb$ to $\N$ or $\Z_2$ depending on whether they correspond to even degree or odd degree parameters. Since the $p_{C_n,\za,\zb}$ are specific $p_{C,D}\,$, they are a countable family of seminorms on $\cO(U)$. To show that the countable family is equivalent to the original one, we use Proposition \ref{Criterion}. For any $p_{C,D}$ there exists $C_n\supset C$, so that $$p_{C,D}(f)=\sup_C|\ze(D(f))|\le \sup_{C_n}|\sum_{\za\zb}\ze(D_{\za\zb}(x,\xi))\,\ze(\p_\xi^\zb\p_x^\za f)|\le $$ \be\label{EquivCDCount} (1+\max_{\za\zb}\sup_{C_n}|\ze(D_{\za\zb}(x,\xi))|)\,\sum_{\za\zb}\sup_{C_n}|\ze(\p_\xi^\zb\p_x^\za f)|=\cC\,\sum_{\za\zb}p_{C_n,\za,\zb}(f)\;.\ee The similar condition for $p_{C,D}$ and $p_{C_n,\za,\zb}$ exchanged is obviously satisfied, since any $p_{C_n,\za,\zb}$ is a specific seminorm of the type $p_{C,D}\,$.\medskip

Finally the $\Z_2^n$-graded vector space $\cO(U)$ is a Fr\'echet space. This should of course mean that $\cO(U)$ is a $\Z_2^n$-graded Fr\'echet vector space in the sense of Definition \ref{FNAS}. In the paragraph following that definition, we mentioned that, if the homogeneous subspaces $\cO^\zg(U)$, $\zg\in\Zn$, are Fr\'echet, then $\cO(U)$ is Fr\'echet as well. A rigorous application of Definition \ref{FNAS} requires now that we take an interest in the converse result. However, what we proved so far is valid for any functions in $\cO(U)$, in particular for the functions of a fixed degree, i.e., for the functions in $\cO^\zg(U)$, $\zg\in\Z_2^n$. It follows that all spaces $\cO^\zg(U)$, $\zg\in\Zn$, are Fr\'echet spaces for the seminorms considered. Hence, the space $\cO(U)$ is a $\Zn$-graded Fr\'echet space in the sense of Definition \ref{FNAS}.\medskip

Alternatively, the reader may observe that any subspace of a Fr\'echet space, which contains the limits of its converging sequences, is itself a Fr\'echet space. Indeed, the restrictions to this subspace of the countable and separating family of seminorms of the total space is again a countable and separating family of seminorms. In view of Proposition \ref{Const}, the resulting locally convex seminorm topology of the subspace is implemented by a translation-invariant metric. To be Fr\'echet, the subspace must still be complete with respect to this metric, i.e., it has to be complete with respect to the seminorm topology. Now, any Cauchy sequence of the subspace is Cauchy in the total space and converges therefore in the total space. But, by assumption, its limit is located in the subspace, so that the subspace is complete with respect to its topology. In the case considered here, any homogeneous subspace $\cO^\zg(U)$ of the Fr\'echet space $\cO(U)$ contains the limits of its converging sequences (in view of point 1 of the preceding lemma) and is thus Fr\'echet (so we can conclude again that $\cO(U)$ is a $\Zn$-graded Fr\'echet vector space in the sense of Definition \ref{FNAS}).\medskip

Of course $\cO(U)$ is equipped with a $\Zn$-graded $\Zn$-commutative associative unital $\R$-algebra structure. Hence $\cO(U)$ is a $\Zn$-graded Fr\'echet algebra, if we can find equivalent countable families of submultiplicative seminorms. We choose the countable family of seminorms defined, for all compacts $C_n$ and all $m,\zm\in\N$, by \be\label{Cnmmu}\zr_{C_n,m,\zm}(f)=2^{m+\zm}\sup_{\!\tiny\begin{array}{c}|\za|\le m\\|\zb|\le\zm\end{array}}\!\!p_{C_n,\za,\zb}(f)=2^{m+\zm}\sup_{\!\tiny\begin{array}{c}|\za|\le m\\|\zb|\le\zm\end{array}}\!\!\sup_{C_n}|\ze(\p_\xi^\zb\p_x^\za f)|\;,\ee see Equation \eqref{Cnalphabeta}. These seminorms are submultiplicative. Indeed, for $f,g\in\cO(U)$, $|\za|\le m$, $|\zb|\le\zm$, and $x\in C_n$, we have, in view of \eqref{EpsDiff}, $$|\ze(\p_\xi^\zb\p_x^\za(f\cdot g))|=|\zb!\,\p_x^\za(f\cdot g)_\zb|=|\zb!\,\p_x^\za\sum_{\zb_1+\zb_2=\zb}f_{\zb_1}\cdot g_{\zb_2}|\le$$ $$\sum_{\zb_1+\zb_2=\zb}\frac{\zb!}{\zb_1!\zb_2!}\,\sum_{\za_1+\za_2=\za}\frac{\za!}{\za_1!\za_2!}\,|\zb_1!\,\p_x^{\za_1}f_{\zb_1}|\cdot |\zb_2!\,\p_x^{\za_2}g_{\zb_2}|\le$$ $$\left(\sum_{\zb_1+\zb_2=\zb}\frac{\zb!}{\zb_1!\zb_2!}\sum_{\za_1+\za_2=\za}\frac{\za!}{\za_1!\za_2!}\right)\sup_{\!\tiny\begin{array}{c}|\za_1|\le m\\|\zb_1|\le\zm\end{array}}\!\!\sup_{C_n}|\ze(\p_\xi^{\zb_1}\p_x^{\za_1}f)|\,\cdot\, \sup_{\!\tiny\begin{array}{c}|\za_2|\le m\\|\zb_2|\le\zm\end{array}}\!\!\sup_{C_n}|\ze(\p_\xi^{\zb_2}\p_x^{\za_2}g)|\le$$ $$2^{m+\zm}\,\sup_{\!\tiny\begin{array}{c}|\za_1|\le m\\|\zb_1|\le\zm\end{array}}\!\!\sup_{C_n}|\ze(\p_\xi^{\zb_1}\p_x^{\za_1}f)|\,\cdot\, \sup_{\!\tiny\begin{array}{c}|\za_2|\le m\\|\zb_2|\le\zm\end{array}}\!\!\sup_{C_n}|\ze(\p_\xi^{\zb_2}\p_x^{\za_2}g)|\;,$$ since the sum over $\zb_1,\zb_2$, for instance, is equal to $2^{|\zb|}$. It follows that $$\zr_{C_n,m,\zm}(f\cdot g)\le\zr_{C_n,m,\zm}(f)\cdot\zr_{C_n,m,\zm}(g)\;.$$ Moreover, the family $\zr_{C_n,m,\zm}$ is equivalent to the family $p_{C_n,\za,\zb}$. On the one hand, we have $$\zr_{C_n,m,\zm}(f)\le 2^{m+\zm}\sum_{\!\tiny\begin{array}{c}|\za|\le m\\|\zb|\le\zm\end{array}}\,p_{C_n,\za,\zb}(f)\;,$$ and on the other, setting $|\za|=m,|\zb|=\zm$, we get $p_{C_n,\za,\zb}(f)\le\zr_{C_n,m,\zm}(f).$\medskip

It remains to show that $\cO(U)$ is a $\Z_2^n$-graded nuclear {\small LCTVS} in the sense of Definition \ref{FNAS}. Since any subspace of a nuclear space is nuclear, it suffices to prove that the locally convex space $\cO(U)$ is nuclear. We set $\mathbf{q}=(\mathbf{q'},\mathbf{q''})$, where $\mathbf{q'}=(q_1,\ldots,q_{2^{n-1}-1})$ and $\mathbf{q''}=(q_{2^{n-1}},\ldots,q_{2^n-1})$ give the number of parameters in each nonzero even degree and the number of parameters in each odd degree, respectively. We also set $\cA=\N^{|\mathbf{q'}|}\times\Z_2^{|\mathbf{q''}|}$. Further, we consider the coordinate order \eqref{CoordOrder} and we order the monomials $\xi^\za=\zy^\zb\zz^\zg$ using the lexicographic order with respect to $\za=(\zb,\zg)$. This leads to a linear vector space isomorphism $$i:\cO(U)\simeq\Ci(U)[[\xi]]\ni\sum_{\za\in \cA} f_\za(x)\xi^\za\mapsto (f_\za)_{\za\in \cA}\in\prod_{\za\in \cA}\Ci(U)\;.$$ We identify $\cO(U)$ with $\prod_{\za\in \cA}\Ci(U)$ via $i$, so that $i=\id$. Since $\Ci(U)$ is nuclear, see Example \ref{ExNuc}, and since any product of nuclear {\small LCTVS}-s is a nuclear {\small LCTVS} for the product topology, the space $\cO(U)$ is nuclear for this topology. It is known that, if $\zp_b:\prod_{a\in A}V_a\to V_b$ is a product of {\small LCTVS}-s $V_a$, whose locally convex topologies are implemented by families of seminorms $(\zr^a_{i})_i$, then the locally convex product topology is implemented by the family of seminorms $(\tilde\zr_i^a)_{a,i}$ defined by $\tilde\zr_i^a=\zr_i^a\circ\zp_a\,$. Hence, in the case considered here, the product topology is given by the family of seminorms $\tilde p^\za_{\zD,C}=p_{\zD,C}\circ\zp_\za\,$. Of course, the standard locally convex topology on $\cO(U)$, i.e., the seminorm topology of the family $p_{C,D}$, must coincide with the product topology, i.e., the families of seminorms $\tilde p^\za_{\zD,C}$ and $p_{C,D}$ must be equivalent. Let therefore $\zb\in\cA$, $f\in\cO(U)$, let $\zD$ be a differential operator acting on $\Ci(U)$, and let $C$ be a compact subset of $U$. When noticing that $D=\frac{1}{\zb!}\p_\xi^\zb\zD$ is a differential operator acting on $\cO(U)$, we obtain $$\tilde p^\zb_{\zD,C}(f)=p_{\zD,C}(\zp_\zb(f))=\sup_C|\zD(f_\zb)|=\sup_C|\ze(\frac{1}{\zb!}\p_\xi^\zb\zD(f))|=p_{C,D}(f)\;.$$ Conversely, for any differential operator $D$ acting on $\cO(U)$, we have $$p_{C,D}(f)=\sup_C|\ze(D(f))|= \sup_{C}|\sum_{\za\zb}\ze(\zb!\,D_{\za\zb}(x,\xi))\,\p_x^\za f_\zb|\le $$ \be\label{EquivCDCount} (1+\max_{\za\zb}\sup_{C}|\ze(\zb!\,D_{\za\zb}(x,\xi))|)\,\sum_{\za\zb}\sup_{C}|\p_x^\za f_\zb|=\cC\,\sum_{\za\zb}\tilde p^\zb_{\p_x^\za,C}(f)\;.\ee\end{proof}

\begin{cor}\label{TVSIsomCor} For any open subset $\zW\subset\R^p$, the map \be\label{TVSIsom}\Ci(\zW)[[\xi]]\ni\sum_{\za\in \cA} f_\za(x)\xi^\za\mapsto (f_\za)_{\za\in \cA}\in\prod_{\za\in \cA}\Ci(\zW)\;,\ee where the source $(\,$resp., the target$\,)$ is equipped with the standard topology induced by $(p_{C,D})_{C,D}$ $(\,$resp., the product topology of the standard topologies induced by $(p_{\zD,C})_{\zD,C}$$\,)$, is an isomorphism of {\small TVS}-s.\end{cor}

The next lemma will allow us to almost complete the proof of Theorem \ref{MainTheo}.

\begin{lem}\label{LEM3} Let $U\subset M$ be an open subset. \begin{itemize}\item Let $(U_i)_{i\in I}$ be an open cover of $U$, let $f_n$, $n\in\N$, and $f$ be $\Zn$-functions in $\cO(U)$. The sequence $f_n$ is a Cauchy sequence in $\cO(U)$ $(\,$resp., converges in $\cO(U)$ to $f$$\,)$ if and only if the sequence $f_n|_{U_i}$ of restrictions is a Cauchy sequence in $\cO(U_i)$ $(\,$resp., converges in $\cO(U_i)$ to the restriction $f|_{U_i}$$\,)$, for all $i\in I$. \item The space $\cO(U)$ is a $\Zn$-graded Fr\'echet algebra.\end{itemize}\end{lem}

\begin{proof} Both statements of Item 1 are of the type \be\label{GlobConv}p_{C,D}(f_n)=\sup_{C}|\ze(D(f_n))|\to 0\ee if $n\to\infty$, for all compact subsets $C\subset U$ and all $\Z_2^n$-differential operators $D$ on $U$, if and only if, for all $i,$ \be\label{LocConv}p_{C_i,D_i}(f_n|_{U_i})=\sup_{C_i}|\ze(D_i(f_n|_{U_i}))|\to 0\ee if $n\to\infty$, for all compact subsets $C_i\subset U_i$ and all $\Zn$-differential operators $D_i$ on $U_i\,$.\medskip

Assume first that \eqref{GlobConv} holds, let $C_i,D_i$ be as said, and consider a bump function $\zg\in\cO(U)$ that equals 1 in an open neighborhood $V$ of $C_i$ and is compactly supported in $U_i$. Since $C_i$ is a compact subset of $U$ and since $\zg D_i$ is a $\Zn$-differential operator on $U$, we get $$p_{C_i,D_i}(f_n|_{U_i})=\sup_{C_i}|\ze(\zg)\ze(D_i(f_n|_{U_i}))|=\sup_{C_i}|\ze((\zg D_i)(f_n))|\to 0\;,$$ if $n\to\infty$. Conversely, if \eqref{LocConv} holds and if $C,D$ are as above, there exists a finite open cover $(V_j)_{j}$ of $C$ such that each $\bar V_j$ is compact and each $\bar V_j$ is contained in some $U_{i(j)}$ \cite{CCF}. Then, $$p_{C,D}(f_n)=\sup_{C}|\ze(D(f_n))|\le \sup_{\cup_j\bar V_j}|\ze(D(f_n))|\le\sum_j\sup_{\bar V_j}|\ze(D|_{U_{i(j)}}(f_n|_{U_{i(j)}}))|\to 0\;,$$ if $n\to\infty$.\medskip

It remains to show that the $\Z_2^n$-graded associative $\R$-algebra $\cO(U)$ is a $\Z_2^n$-graded Fr\'echet algebra for its standard Hausdorff locally convex topology given by the family $(p_{C,D})_{C,D}$, see Lemma \ref{LEM1}. Note that the space $\cO(U)$ is complete. Indeed, let $U_i$ be a countable open cover of $U$ by $\Zn$-chart domains and let $f_n$ be a Cauchy sequence in $\cO(U)$. In view of Item 1 and Lemma \ref{LEM2}, the sequence $f_n|_{U_i}$ converges to $f_{U_i}$ in $\cO(U_i)$. When applying Item 1 to the open cover $U_{ij}=U_i\cap U_j$ of $U_i$, to the sequence $f_n|_{U_i}$, and the function $f_{U_i}$, we find that $f_n|_{U_{ij}}$ converges in $\cO(U_{ij})$ to $f_{U_i}|_{U_{ij}}$. One sees similarly that it converges also to $f_{U_j}|_{U_{ij}}$. Hence, there is a unique function $f\in\cO(U)$ such that $f|_{U_i}=f_{U_i}$. It now follows from Item 1 that $f_n\to f$ in $\cO(U)$, so that $\cO(U)$ is complete. To conclude that $\cO(U)$ is a $\Zn$-graded Fr\'echet algebra, we have to provide an equivalent countable submultiplicative family of seminorms on $\cO(U)$. It actually suffices to proceed as above, see \eqref{Cnmmu}. More precisely, for each one of the countably many $\Zn$-chart domains $U_i$, we can choose a countable cover of $U_i$ by compact subsets $C_{n,i}\subset U_i$ such that $C_{n,i}$ is contained in the interior of $C_{n+1,i}\,$. The family $$\zr_{C_{n,i},m,\zm}(f)=2^{m+\zm}\sup_{\!\tiny\begin{array}{c}|\za|\le m\\|\zb|\le\zm\end{array}}\!\!\sup_{C_{n,i}}|\ze(\p_\xi^\zb\p_x^\za f|_{U_i})|\;,$$ where $m,\zm\in\N$, is the searched countable equivalent submultiplicative family of seminorms on $\cO(U)$. \end{proof}

Theorem \ref{MainTheo} can now be proved as follows. Let $U$ be any open subset of $M$ and $(U_i)_{i\in I}$ any open cover of $U$. In view of Item 1 of Lemma \ref{LEM3}, the restriction maps $\cO(U)\ni f\mapsto f|_{U_i}\in\cO(U_i)$ are sequentially continuous. Since a map from a metrizable {\small TVS} into a {\small TVS} is continuous if and only if it is sequentially continuous \cite{Tre}, the preceding restrictions are continuous (in particular any restriction is continuous). The fact that $\cO(U)$, $U\subset M$, carries the coarsest or initial topology with respect to these restrictions is a consequence of the open mapping theorem for Fr\'echet spaces and of the fact that $M$ is second countable \cite{Mal}. In view of this initial topology property, the Fr\'echet sheaf $\cO$ of $\Zn$-graded algebras is nuclear if $\cO(U)$ is nuclear for any open $U\subset M$, or, equivalently, if $\cO(U)$ is nuclear for any $U$ of an open basis of $M$ \cite{Mal}. As mentioned above, any open subset of $M$ is a union of open $\Zn$-chart domains, so that it suffices that $\cO(U)$ be nuclear for any $\Zn$-chart domain $U$, which has been proven in Lemma \ref{LEM2}. \end{proof}

We mentioned above that any $\Zn$-differential operator $D:\cO\to\cO$ and any $\Zn$-morphism $\Phi=(\zf,\zf^*):(M,\cO_M)\to (N,\cO_N)$ are $\cJ$-adically continuous. The adic continuity allowed us to conclude that differential operators and pullbacks act on series by acting on each of their terms. The next proposition states that differential operators and pullbacks are continuous with respect to the standard locally convex topology of $\cO$. This property will be of importance later on.

\begin{theo}\label{Cont} Let $\cM=(M,\cO_M)$ and ${\cN}=(N,\cO_N)$ be $\Zn$-manifolds, let $D$ be a $\Zn$-differential operator on $\cM$ and let $\Phi=(\zf,\zf^*):\cM\to\cN$ be a $\Zn$-morphism. \begin{itemize} \item For any open $U\subset M$, the differential operator $D_U:\cO_M(U)\to\cO_M(U)$ is continuous for the standard locally convex topology of $\cO_M(U)$. \item For any open $V\subset N$, the pullback $\zf^*_V:\cO_N(V)\to\cO_M(\zf^{-1}(V))$ is continuous for the standard locally convex topologies of $\cO_N(V)$ and $\cO_M(\zf^{-1}(V))$.
\end{itemize}\end{theo}

\begin{proof} Since in both cases the source {\small TVS} is metrizable, it suffices to prove that the linear maps $D_U$ and $\phi^*_N$ are sequentially continuous. Item 1 is obvious in view of Equation \eqref{StandSemiNormTop}.\medskip

We prove Item 2 first in the case $V=N$. Let $f_n\to 0$ in $\cO(N)$ and show that $\phi^*f_n\to 0$ in $\cO(M)$ (we omit subscripts $N$ and $M$). It follows from Item 1 of Lemma \ref{LEM3} that both convergences are equivalent to the convergences to 0 of the restrictions of $f_n$ and $\phi^*f_n$ to an arbitrary open cover of $N$ and $M$, respectively. Let $N=\cup_jV_j$ be an open cover of $N$ by $\Zn$-chart domains of $\cN$. We thus get $M=\cup_j\,\phi^{-1}(V_j)$ an open cover of $M$. Each $\phi^{-1}(V_j)$ can again be covered by $\Zn$-chart domains $U_{j,i}$ of $\cM$. Since $M=\cup_{ji}\,U_{j,i}$, it suffices to prove that $(\phi^*f_n)|_{U_{j,i}}\to 0$ in $\cO(U_{j,i})$, knowing that $f_n|_{V_j}\to 0$ in $\cO(V_j)$. We will omit the subscripts $j,i$ in the arbitrarily chosen $U_{j,i}\subset \phi^{-1}(V_j)$, as well as subscript $j$ in $V_j$, and we will write $\zn_n$ for the restriction $f_n|_V$. We denote by $u=(u^A)=(x^a,\xi^\frak{a})$ the coordinates in $U$ and by $v=(v^B)=(y^b,\zh^\frak{b})$ the coordinates in $V$. Let now $C\subset U$ be a compact subset and let $$D=\sum_\za D_\za(u)\p_u^\za$$ be a differential operator acting on $\cO(U)$. We have to prove that
\be\label{SNConv}
p_{C,D}((\phi^*f_n)|_U)=\sup_C|\ze(D((\phi^*f_n)|_U))|=\sup_C|\ze\sum_\za D_\za(u)\,\p_u^\za(\phi^*\zn_n)|_U|\to 0\;.
\ee
The $\Zn$ chain rule \cite{CKP}
$$
\partial_{u^A} (\phi^*f) = \sum_B \partial_{u^A} (\phi^*v^B) \, \phi^* (\partial_{v^B} f)
$$
extends to the $\Zn$ Fa\`a di Bruno formula
$$
\p_u^\za(\phi^*f) = \sum\left(\prod \cC\,\p_u^\zb(\phi^*v^B)\right)\,\phi^*(\p_v^\zg f)\;,
$$
where the sum and products are finite, where $\cC$ denotes real numbers, and where we limited ourselves to the structure of this complex result. It follows that the supremum in \eqref{SNConv} reads
\be\label{SNConv2}
\sup_C\left|\ze\sum_\za \sum\left(D_\za(u)\prod \cC\,\p_u^\zb(\phi^*v^B)\right)\,\phi^*(\p_v^\zg \zn_n)\right|=\sup_C\left|\ze\sum F(u)\,\phi^*(\p_v^\zg \zn_n)\right|=\sup_C|-|\;,
\ee
where we omitted the restrictions to $U$ and where $F\in\cO(U)$. We get
$$
\sup_C|-|=\sup_C\left|\sum (\ze\, F)(x)\,\ze(\p_v^\zg \zn_n)(\phi(x))\right|\le
$$
$$
\sum \sup_C|\ze F|\,\sup_C|\ze(\p_v^\zg \zn_n)(\phi(x))|=\sum \sup_C|\ze F|\,\sup_{\phi(C)}|\ze(\p_v^\zg \zn_n)|\;.
$$
Since $\zn_n\to 0$ in $\cO(V)$ and $\phi(C)$ is a compact subset of $V$, the conclusion follows.\medskip

In order to extend the conclusion from $\zf^*_N$ to $\zf^*_V$, for any open $V\subset N$, note that we can restrict the $\Zn$-morphism $$(\zf,\zf^*):(M,\cO_M)\to(N,\cO_N)$$ to a $\Zn$-morphism $(\zvf,\zvf^*)$ between the open $\Zn$-submanifolds $(U,\cO_M|_U)$ and $(V,\cO_N|_V)$, where $U=\zf^{-1}(V)$. It suffices to set $\zvf=\zf|_U:U\to V$, and to set, for any open $W\subset V$, $$\zvf^*_W=\zr^{\zf^{-1}(W)}_{U\cap\,\zf^{-1}(W)}\circ\zf^*_W:\cO_N(W)\to \cO_M(\zvf^{-1}(W))\;.$$ Indeed, the base map $\zvf$ is smooth and the pullbacks $\zvf^*_W$ are $\Zn$-graded unital $\R$-algebra morphisms, which commute with restrictions. Applying now the first part of our proof of Item 2, we get that $\zvf^*_V:\cO_N(V)\to \cO_M(\zf^{-1}(V))$ is continuous. This concludes the proof, since $\zvf^*_V=\zf^*_V$.
\end{proof}

\section{Appendix}

In this section, we recall basic results and provide examples.

\subsection{Fr\'echet spaces}

\subsubsection{Definitions and construction}

\begin{rem} All vector spaces considered in the present text are spaces over the field $\R$ of real numbers.\end{rem}

\begin{defi} A topological vector space ({\small TVS}) is \emph{locally convex} if its topology has a basis made of convex subsets, i.e., subsets $U$ such that, for any $x,y\in U$, the segment $\{(1-t)x+ty:t\in[0,1]\}$ is contained in $U$.
\end{defi}

\begin{defi}\label{Frechet} A locally convex topological vector space ({\small LCTVS}) is a \emph{Fr\'echet space} if its topology can be implemented by a translation-invariant metric with respect to which it is (sequentially) complete.
\end{defi}

Recall that a metrizable {\small TVS} is complete if and only if it is sequentially complete.\medskip

The standard construction of a Fr\'echet space starts from a family of seminorms. The difference between a seminorm $p$ on a vector space and a norm $||-||$ is that $p(x)=0$ does not imply that $x=0$. Recall also that a family of seminorms $(p_i)_{i\in I}$ \emph{separates points} (or is \emph{separating}), if for $x\neq 0$, there is $i\in I$ such that $p_i(x)\neq 0$.\medskip

The following proposition is almost obvious and will not be proven.

\begin{prop} Let $(p_i)_{i\in I}$ be a family of seminorms on a vector space $V$. For any $i\in I,x\in V,\ze>0$, set $$B_i(x,\ze)=\{y\in V:p_i(y-x)<\ze\}\;.$$ The family of all the subsets $B_i(x,\ze)$ generates a topology on $V$ $(\,$recall that this topology is made of the unions of finite intersections of subsets $B_i(x,\ze)$$\,)$. We refer to this topology as the \emph{seminorm topology} induced by $(p_i)_{i\in I}$. The finite intersections of subsets $B_i(x,\ze)$ \emph{with fixed $x$ and $\ze$}, form a basis $\cB$ of the seminorm topology: $$\cB=\{\cap_{j=0}^nB_{i_j}(x,\ze):n\in\N, i_j\in I, x\in V, \ze>0\}\;.$$ The elements of this basis are convex. The open subsets $U$ of the seminorm topology are characterized by the property that for any $x\in U$, there is a basis element $\cap_{j=0}^nB_{i_j}(x,\ze)$, which is contained in $U$. The seminorm topology endows $V$ with a structure of {\small LCTVS}.\end{prop}

\begin{prop}\label{Equiv} Let $(p_i)_{i\in I}$ be a family of seminorms on a vector space $V$. A sequence $(x_k)_{k\in\N}$ of elements of $V$ converges to $x\in V$ with respect to the seminorm topology, if and only if it converges to $x$ with respect to each seminorm, i.e., if and only if $$\lim_kp_i(x_k-x)= 0,\;\forall i\in I\;.$$ Similarly, the sequence $(x_k)_{k\in\N}$ is a Cauchy sequence with respect to the seminorm topology, if and only if it is Cauchy with respect to each seminorm. \end{prop}

\begin{proof} We prove the second statement. Let $(x_k)_{k\in\N}$ be Cauchy with respect to the topology and take $i\in I$ and $\ze>0$. Since $B_i(0,\ze)$ is an open neighborhood of $0$, there is $N\ge 0$ such that $x_\ell-x_m\in B_i(0,\ze)$ if $\ell,m > N$. It follows that $p_i(x_\ell-x_m)<\ze$ if $\ell,m > N$. Conversely, let $U$ be an open neighborhood of $0$ and let $\cap_{j=0}^n B_{i_j}(0,\ze)\subset U$. The assumption implies that, for all $j\in\{0,\ldots,n\}$, there is $N_j\ge 0$ such that $p_{i_j}(x_\ell-x_m)<\ze$ if $\ell,m > N_j$. Hence, if $\ell,m > N=\op{sup}_jN_j$, the difference $x_\ell-x_m$ is in $U$.\end{proof}

Further:

\begin{prop} A seminorm topology is Hausdorff if and only if its inducing family of seminorms is separating.\end{prop}

\begin{proof} Let $(p_i)_{i\in I}$ be the family of seminorms on the vector space $V$. Assume first that the seminorm topology is Hausdorff and let $x$ be a non-zero vector in $V$. In view of the Hausdorff property, there is a neighborhood of $x$ which does not contain $0$. Hence, there is an open subset $B_i(x,\ze)$ which does not contain $0$: $p_i(x)\ge \ze$ and $p_i(x)\neq 0$. Conversely, if the family of seminorms separates points and if $x,y$ are two different vectors in $V$, there exists $i\in I$ such that $p_i(x-y)\neq 0$, i.e., such that $p_i(x-y)=\zh>0$. Let now $\ze=\zh/2>0$ and take the open neighborhoods $B_i(x,\ze)$ of $x$ and $B_i(y,\ze)$ of $y$. If these neighborhoods have a common vector $z$, then $$2\ze=p_i(x-y)\le p_i(x-z)+p_i(z-y)<2\ze\;,$$ so that the neighborhoods are disjoint.  \end{proof}

As well-known, the next definition of Fr\'echet spaces is equivalent to the above one, but is better suited for applications.

\begin{defi}\label{Frechet2} A {\small TVS} is a \emph{Fr\'echet space} if it is Hausdorff and (sequentially) complete, and if its topology can be induced by a countable family of seminorms.\end{defi}

We are now prepared to discuss standard construction methods of Fr\'echet spaces from countable families of seminorms.\medskip

\begin{prop}\label{Const} If $(p_n)_{n\in\N}$ is a countable family of seminorms on a vector space $V$, and if this family separates points, then \be\label{TraInvComMet}d(x,y):=\sum_{n\in\N}\frac{1}{2^n}\frac{p_n(x-y)}{1+p_n(x-y)},\quad x,y\in V\;,\ee is a translation-invariant metric on $V$ that induces the seminorm topology of $V$.
\end{prop}

\begin{proof} The statement is a standard functional analytical result.\end{proof}

The next proposition is natural. It extends Proposition \ref{Equiv}:

\begin{prop}\label{Compl} In the situation of Proposition \ref{Const}, a sequence in $V$ converges to a limit in $V$ (resp., is a Cauchy sequence) with respect to the metric $d$, if and only if it converges (resp., is a Cauchy sequence) with respect to the seminorm topology, and if and only if it converges (resp., is a Cauchy sequence) with respect to all seminorms.\end{prop}

\begin{rem}\label{Method} To construct a Fr\'echet space, one usually proceeds as follows, although the method admits a number of variants. On starts with a vector space $V$ and a countable and separating family $(p_n)_{n\in\N}$ of seminorms on it. The seminorm topology turns $V$ into a Hausdorff {\small LCTVS}. Definition \ref{Frechet2} then allows us to conclude that $V$ is a Fr\'echet space, if we can prove that $V$ is (sequentially) complete with respect to its topology. In view of Proposition \ref{Compl}, this condition is equivalent to (sequential) completeness with respect to the translation-invariant metric \eqref{TraInvComMet} induced by the seminorms. Further, if one can verify that a sequence in $V$ that is Cauchy for any seminorm $p_n$, converges to a fixed $x\in V$ for any $p_n$, then the {\small TVS} $V$ is (sequentially) complete with respect to the seminorm topology, again due to Proposition \ref{Compl}.\end{rem}

\subsubsection{Examples}

We briefly present some examples.

\begin{ex}The vector space $\R^\infty$ of all sequences $r=(r_0,r_1,\ldots)$ of real numbers is a Fr\'echet space for the countable family $(p_n)_{n\in\N}$ of seminorms $$p_n(r)=\sup_{m\le n}|r_m|\;.$$\end{ex}

The family $(q_n)_{n\in\N}$ given by $$q_n(r)=\sum_{r\le n}|r_m|\;$$ defines the same topology, i.e., the same Fr\'echet space. We say that two such families of seminorms are {\em equivalent}.\medskip

We recall an important criterion for equivalence of two families of seminorms.

\begin{prop}\label{Criterion} To families of seminorms $(p_i)_{i\in I}$ and $(q_j)_{j\in J}$ on a vector space $V$ are equivalent, i.e., they induce the same locally convex topology, if and only if, for any $i$, there is a constant $C>0$ and a finite subset $\{j_1,\ldots,j_N\}\subset J$, such that $$p_i(x)\le C(q_{j_1}(x)+\ldots+q_{j_N}(x)),\;\forall x\in V\;,$$ and vice versa.\end{prop}

The next observation is indispensable.

\begin{lem}\label{CompCov} Any open subset of a (second-countable Hausdorff finite-dimensional) smooth manifold admits a countable cover by compact subsets.\end{lem}

\begin{proof} Every second-countable topological space is a Lindel\"of space, i.e., any open cover admits a countable subcover. Since all smooth manifolds considered in our texts on $\Z_2^n$-Geometry are second-countable Hausdorff finite-dimensional smooth manifolds, the Lindel\"of property holds for all open subsets $U\subset M$.\medskip

Any open subset $\zW\subset\R^p$ admits a cover $\cup_{x\in \zW} B(x)$ by open balls $B(x)$ whose adherence $\bar B(x)$ is contained in $\zW$. In view of the Lindel\"of property, we can extract from the preceding open cover of $\zW$ a countable subcover \be\label{CCCS}\zW=\cup_{i\in\N}B(x_i)=\cup_{i\in\N}\bar B(x_i)\;.\ee The latter cover is searched countable cover by compact subsets.\medskip

Let now $U\subset M$ be an open subset of a $p$-dimensional smooth manifold. We can cover $U$ by coordinate systems $(U_\za,\zvf_\za)$ and extract a countable subcover $(U_i,\zvf_i)$. Since $\zvf_i$ is a homeomorphism between $U_i\subset M$ and $\zvf_i(U_i)\subset\R^p$, the set $U_i$ admits a countable cover $U_i=\cup_{k\in\N}C_{ki}$ by compact subsets $C_{ki}\subset U_i$. We thus get the countable cover $U=\cup_{i\in\N}\cup_{k\in\N}C_{ki}$ of $U$ by compact subsets $C_{ki}\subset U$.\end{proof}

The following is one of the important examples of Fr\'echet spaces.

\begin{ex} For any open subset $\zW\subset \R^p$, the function algebra $\Ci(\zW)$ is a Fr\'echet vector space for the countable family $(p_{\za,i})_{\za,i}$ of seminorms defined, for any multi-index $\za\in\N^{\times p}$ and any compact $C_i$ ($i\in\N$) of a countable cover of $\zW$ by compact subsets (e.g., $C_i$ may run through the balls $\bar B(x_i)$ of (\ref{CCCS})), by \be\label{SNCi} p_{\za,i}(f)=\sup_{x\in C_i}|\p_{x}^\za f|\;.\ee\end{ex}

\noindent To prove this standard statement one uses Remark \ref{Method}. The result can be extended:

\begin{ex}\label{FrechetSmoothFctBase} The function algebra $\Ci(U)$ of an open subset $U$ of a (second-countable Hausdorff finite-dimensional) smooth manifold $M$, is a Fr\'echet vector space. The locally convex topology of $\Ci(U)$ is implemented (for instance) by the family of seminorms \be\label{SNCiU}p_{\zD,C}(f)=\sup_{x\in C}|\zD(f)(x)|\;,\ee where $\zD\in\cD(U)$ is any differential operator acting on $\Ci(U)$ and where $C$ is any compact subset of $U$. Note that -- by definition -- this topology is the topology of uniform convergence on compact subsets $C$ of $f$ and its `derivatives' $\zD(f)$.\end{ex}

\subsection{Nuclear spaces}

\subsubsection{Definition}\label{NuclearDefinition}

Let us recall that a {\it completion} of a {\small TVS} $V$ is a complete {\small TVS} $\hat V$ that contains $V$ (or, better, a homeomorphic image of $V$) as a dense subspace. Any {\small (LC)TVS} can be completed as {\small (LC)TVS}.\medskip

When considering (algebraic) tensor products of {\small LCTVS}-s, some subtleties arise due to the possibility to choose various topologies on these products.\medskip

More precisely, let $V,W$ be two {\small LCTVS}-s. The finest locally convex topology on the algebraic tensor product $V\otimes W$, for which the natural map $V\times W\to V\otimes W$ is continuous, is referred to as the \emph{projective tensor topology}. The completion of the resulting {\small LCTVS} is the completed projective tensor product $V\widehat\otimes_\zp W$. There exists another natural locally convex topology on $V\otimes W$, which is coarser than the projective one, and which is called the injective tensor topology. The corresponding completion is the completed injective tensor product $V\widehat\otimes_i W$. Any reasonable locally convex topology on $V\otimes W$ lies between the injective and projective ones.\medskip

We are now prepared to give one of the equivalent definitions of nuclear {\small LCTVS}-s.

\begin{defi} A {\small LCTVS} $V$ is \emph{nuclear} if, for any {\small LCTVS} $W$, the canonical map $V\widehat\otimes_\zp W\to V\widehat\otimes_i W$ is an isomorphism of {\small LCTVS}-s.\end{defi}

More precisely, the identity $\id: V\0_\zp W\to V\0_i W$ is a bijective continuous linear map, and its continuous extension $\widehat{\id}:V\widehat{\0}_\zp W\to V\widehat{\0}_i W$ is an injective continuous linear map. When $V$ is nuclear, this canonical map is onto, or, better, it is a {\small TVS}-isomorphism. As already said, any (reasonable) locally convex topology on $V\0 W$ is located between the projective and the injective tensor topologies. Hence, if $V$ is nuclear, the complete {\small TVS} $V\widehat{\0}W$ is independent of the locally convex topology considered.\medskip

\emph{Nuclear Fr\'echet spaces} are just a specific type of nuclear {\small LCTVS}-s. Fr\'echet spaces are a full subcategory of {\small TVS}-s, and so are nuclear spaces.

\subsubsection{Example}\label{ExNuc}

When thinking about the duality between spaces and function algebras, one meets the problem of interpreting a tensor product of function algebras as function algebra of some space. Even in the case of algebras $\Ci(\zW)$ of smooth functions on open subsets $\zW$ of Euclidean spaces, the canonical map $\Ci(\zW')\0\Ci(\zW'')\to\Ci(\zW'\times\zW'')$ is (of course) not an isomorphism. However, if one endows the algebraic tensor product of the {\small LCTVS}-s $\Ci(\zW')$ and $\Ci(\zW'')$ with the projective tensor topology and considers the corresponding completion, one gets an isomorphism of {\small TVS}-s: \be\label{CompProdSmooth}\Ci(\zW')\widehat\0_\zp\Ci(\zW'')\simeq\Ci(\zW'\times\zW'')\;.\ee The topology that we choose on the algebraic tensor product is actually irrelevant -- a space of the type $\Ci(\zW)$ is nuclear. More precisely, both, $\Ci(\zW)$ ($\zW\subset\R^p$) and $\Ci(U)$ ($U\subset M$, $M$ smooth manifold), are nuclear Fr\'echet spaces.\medskip

In more detail, if $V,W$ are complete {\small LCTVS}-s and if $V$ is a concrete space (e.g., $V=\Ci(\zW')$), then it is mostly impossible to characterize both $V\widehat\0_\zp W$ and $V\widehat\0_i W$ concretely (in fact a space of bilinear forms on dual spaces of $V$ and $W$ is also involved here, but we refrain from describing this space precisely). For example, when $V=\Ci(\zW')$ and $W=\Ci(\zW'')$, we can interpret $\Ci(\zW')\widehat\0_i\Ci(\zW'')$ concretely as the space $\Ci(\zW'\times \zW'')$, but we have no good access to $\Ci(\zW')\widehat\0_\zp\Ci(\zW'')$. If we know a priori that $V=\Ci(\zW')$ is nuclear, the problem disappears.\medskip

\subsection{Fr\'echet algebras}

In fact, the algebra $\Ci(U)$, where $U$ is any open subset of any smooth manifold, is a Fr\'echet algebra \cite{TopAlg}. The definition of a Fr\'echet algebra is natural:

\begin{defi} A \emph{Fr\'echet algebra} is a Fr\'echet vector space $A$, which is equipped with an associative bilinear and (jointly) continuous multiplication $\cdot:A\times A\to A$. If $(p_i)_{i\in I}$ is a family of seminorms that induces the topology on $A$, (joint) continuity is equivalent to the existence, for any $i\in I$, of $j\in I$, $k\in I$, and $C>0$, such that $$p_i(x\cdot y)\le C\, p_j(x)\, p_k(y),\;\forall x,y\in A.$$ We can always consider an equivalent increasing countable family of seminorms $(||-||_n)_{n\in\N}$. The preceding condition then requires that, for any $n\in\N$, there is $r\in\N, r\ge n$ and $C>0$, such that $$||x\cdot y||_n\le C\, ||x||_r\, ||y||_r,\;\forall x,y\in A.$$ In particular, the topology can be induced by a countable family of submultiplicative seminorms, i.e., by a family  $(p_n)_{n\in\N}$ that satisfies $$p_n(x\cdot y)\le p_n(x)\,p_n(y),\;\forall n\in\N, \forall x,y\in A.$$ \end{defi}

Note that many authors define a Fr\'echet algebra simply as a Fr\'echet vector space, which carries an associative bilinear multiplication, and whose topology can be induced by a countable family $(q_n)_{n\in\N}$ of submultiplicative seminorms. This latter definition is equivalent to the former.

\subsection{Fr\'echet sheaves}

Let $M$ be a smooth manifold and denote by ${\tt Open}(M)$ (resp., ${\tt Alg}(\R)$) the category of open subsets of $M$ (resp., of associative $\R$-algebras). As mentioned above, the function sheaf $$\Ci:{\tt Open}(M)^{\op{op}}\ni U\mapsto \Ci(U)\in{\tt Alg(\R)}$$ is actually valued in \emph{nuclear Fr\'echet algebras}, i.e., in nuclear Fr\'echet vector spaces that carry a Fr\'echet algebra structure. In view of this observation, it is natural to consider Fr\'echet sheaves. Their definition is well-known:

\begin{defi}\label{FSh} A sheaf $\cF$ of (real) vector spaces over a smooth manifold $M$ is a \emph{Fr\'echet sheaf of vector spaces}, if the next two conditions are satisfied: \begin{itemize}\item for any $U\in{\tt Open}(M)$, the vector space $\cF(U)$ is a Fr\'echet vector space, and \item for any $U\in{\tt Open}(M)$ and any cover $(U_i)_{i\in I}$ of $U$, $U_i\in{\tt Open}(U)$, the locally convex topology on $\cF(U)$ is the coarsest topology for which the restriction maps $\cF(U)\to\cF(U_i)$ are continuous.\end{itemize}\end{defi}

Since $\cF$ is a sheaf of (real) vector spaces, it follows from the second condition that, for any $V\in{\tt Open}(U)$, the restriction map $\cF(U)\to\cF(V)$ is $\R$-linear and continuous. As Fr\'echet spaces are a full subcategory of {\small TVS}-s, the second requirement of Definition \ref{FSh} is thus quite natural. In view of this understanding, it is clear that the definition of a \emph{(nuclear) Fr\'echet sheaf of algebras} is similar, but starts from a sheaf of (real) algebras.

\end{document}